\documentclass{article}
\usepackage{preamble}

\begin{document}

\maketitle

\begin{abstract}
    Prediction markets are useful for estimating probabilities of claims whose truth will be revealed at some fixed time -- this includes questions about the values of real-world events (i.e. statistical uncertainty), and questions about the values of primitive recursive functions (i.e. logical or algorithmic uncertainty). However, they cannot be directly applied to questions without a fixed resolution criterion, and real-world applications of prediction markets to such questions often amount to predicting not whether a sentence is true, but whether it will be proven. Such questions could be represented by countable unions or intersections of more basic events, or as First-Order-Logic sentences on the Arithmetical Hierarchy (or even beyond FOL, as hyperarithmetical sentences). In this paper, we propose an approach to betting on such events via options, or equivalently as bets on the outcome of a ``verification-falsification game''. Our work thus acts as an alternative to the existing framework of Garrabrant induction for logical uncertainty, and relates to the stance known as \emph{constructivism} in the philosophy of mathematics; furthermore it has broader implications for philosophy and mathematical logic.
\end{abstract}

\section{Introduction}

Prediction markets are traditionally used for eliciting probabilities on claims whose truth will be revealed at some fixed point \cite{conitzerPredictionMarketsMechanism2012, hansonLogarithmicMarketScoring2002, hansonCombinatorialInformationMarket2003}; they can also straightforwardly be extended to \emph{verifiable claims} (claims that if true, will be revealed to be true) and \emph{falsifiable claims} (claims that if false, will be revealed to be false). For example, ``Bob will die'' is verifiable, and the price of the asset that pays \$1 if Bob dies reflects the probability that he will die; ``Bob is immortal'' is falsifiable, and the price of the asset that immediately pays \$1 which must be returned if Bob dies reflects the probability that he is immortal.\footnote{Such an asset may be traded on a market as defined by standard mechanisms such as continuous double auctions or logarithmic market scoring. For the details of such algorithms see \cite{conitzerPredictionMarketsMechanism2012, hansonLogarithmicMarketScoring2002, hansonCombinatorialInformationMarket2003} or real-life prediction markets such as PredictIt, SMarkets, or the play-money market Manifold.}

Although falsifiability occupies a central role in the imagination of philosophers as to what constitutes a meaningful sentence \cite{thorntonKarlPopper2023}, a much wider class of sentences are treated in science and in common practice generally; such classes of sentences are described by various \emph{logics}. In this paper we will be concerned with designing prediction markets for sentences in \emph{first-order logic} (FOL), though our results may also be straightforwardly generalized to a hyperarithmetical logic up to any computable ordinal (see e.g. \cite{pohlersComputabilityTheoryHyperarithmetical1999, ashComputableStructuresHyperarithmetical2000, moschovakisDescriptiveSetTheory2009, kechrisClassicalDescriptiveSet1995} for standard introductions); Fig~\ref{fig:fol} shows the ``arithmetical hierarchy'' of First-Order Logic sentences.

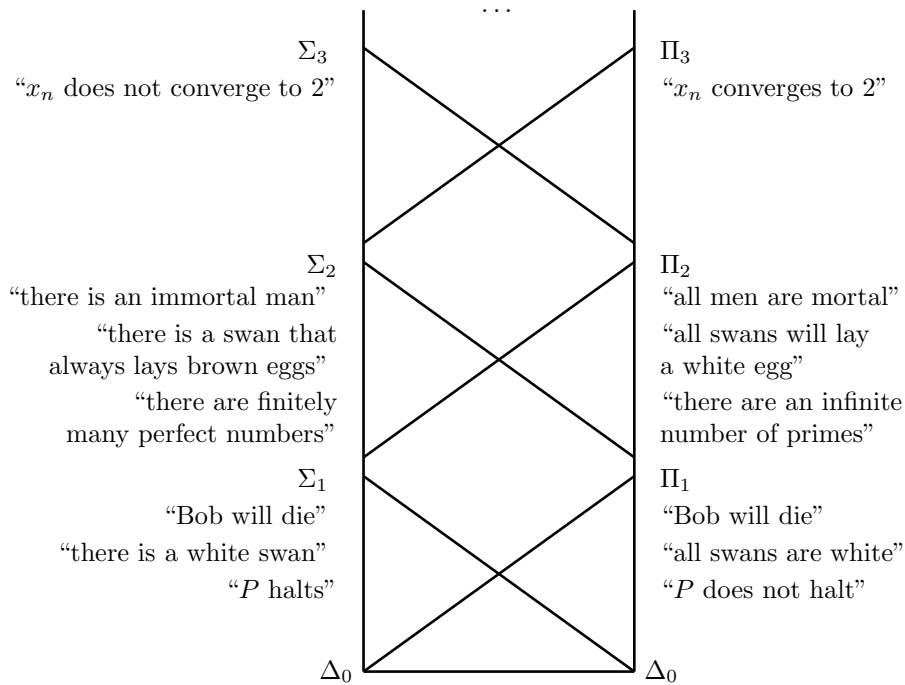
\begin{figure}
    \centering
    \usetikzlibrary{arrows}
\usetikzlibrary{positioning}
\pgfmathsetmacro{\myw}{3.6}
\pgfmathsetmacro{\myh}{2.6}
\pgfmathsetmacro{\mys}{0.25}
\newlength{\mypad}
\setlength{\mypad}{8pt}
\newlength{\mylsg}
\setlength{\mylsg}{2pt}

\begin{tikzpicture}[line width =1pt]

\draw (0,0) -- (0,3*\myh+4*\mys);
\draw (\myw,0) -- (\myw,3*\myh+4*\mys);
\node[] at (\myw/2,3*\myh+4*\mys){$\dots$};
\draw node[above = \mypad, anchor = north east]{$\Delta_0$} (0,0) -- (\myw,0) node[above = \mypad, anchor = north west]{$\Delta_0$};
\draw (0,0) -- (\myw,\myh) node[above = \mypad, anchor = north west]{
    \begin{tabular}{l}
        $\Pi_1$ \vspace{\mylsg}\\
        ``Bob will die'' \vspace{\mylsg}\\
        ``all swans are white'' \vspace{\mylsg}\\
        ``$P$ does not halt''
    \end{tabular}};
\draw (\myw,0) -- (0,\myh) node[above = \mypad, anchor = north east]{
    \begin{tabular}{r}
        $\Sigma_1$ \vspace{\mylsg}\\
        ``Bob will die'' \vspace{\mylsg}\\
        ``there is a white swan'' \vspace{\mylsg}\\
        ``$P$ halts''
    \end{tabular}};
\draw (0,\myh+\mys) -- (\myw,2*\myh+\mys) node[above = \mypad, anchor = north west]{
    \begin{tabular}{l}
        $\Pi_2$ \\
        ``all men are mortal'' \vspace{\mylsg}\\
        ``all swans will lay\\
        a white egg'' \vspace{\mylsg}\\
        ``there are an infinite\\
        number of primes''
    \end{tabular}};
\draw (\myw,\myh+\mys) -- (0,2*\myh+\mys) node[above = \mypad, anchor = north east]{
    \begin{tabular}{r}
        $\Sigma_2$ \\
        ``there is an immortal man'' \vspace{\mylsg}\\
        ``there is a swan that\\
        always lays brown eggs'' \vspace{\mylsg}\\
        ``there are finitely\\
        many perfect numbers''
    \end{tabular}};
\draw (0,2*\myh+2*\mys) -- (\myw,3*\myh+2*\mys) node[above = \mypad, anchor = north west]{
    \begin{tabular}{l}
        $\Pi_3$ \vspace{\mylsg}\\
        ``$x_n$ converges to 2''
    \end{tabular}};
\draw (\myw,2*\myh+2*\mys) -- (0,3*\myh+2*\mys) node[above = \mypad, anchor = north east]{
    \begin{tabular}{r}
        $\Sigma_3$ \vspace{\mylsg}\\
        ``$x_n$ does not converge to 2''
    \end{tabular}};
\end{tikzpicture}
    \caption{The ``arithmetical hierarchy'' of FOL sentences. A $\Sigma_{n+1}$ sentence is of the form $\exists x,p(x)$ where $p$ is $\Pi_n$; a $\Pi_{n+1}$ sentence is of the form $\forall x,p(x)$ where $p$ is $\Sigma_n$. In logic, the lowest level of the arithmetical hierarchy $\Delta_0=\Sigma_0=\Pi_0$ are sentences of the form $f(n)=0$ where $f$ is a primitive recursive function, but this formalism may be extended to empirical truths that will be revealed at a fixed time.}
    \label{fig:fol}
\end{figure}

We will first, in Sections 1.1--1.3, discuss three apparent approaches to our problem, and explain why they fail or are inapplicable to us. In Section~\ref{sec:framework} we will then present our own framework, and in Section~\ref{sec:results} we will prove some basic results pertaining to it.

\subsection{Garrabrant induction}
\label{sec:garrabrant}

A related problem is that of assigning probabilities to sentences of formal mathematics, i.e. \emph{logical uncertainty}. While this neither subsumes nor is subsumed by our problem (there are non-FOL mathematical theories, and conversely we may consider FOL sentences over e.g. empirical events rather than a formal language), there is substantial overlap, and it is worth considering if the methods in that area can be transferred to our problem. The crown jewel of the logical uncertainty literature, \emph{Garrabrant induction} \cite{garrabrantLogicalInduction2016}, in fact implements a prediction market for all logical sentences which pays off whenever a sentence is proven by some formal theorem prover (which exists for any ``computably enumerable theoory'', including first-order arithmetic etc.). However, this would simply be a prediction market for the \emph{provability} of a sentence (in some particular formal theory); the actual Garrabrant induction framework instead adopts a notion of ``propositionally consistent worlds'' that allows even unprovable (and provably unprovable) sentences to have non-zero probabilities and such that basic probabilistic laws like $\Prob(P)+\Prob(\lnot P)=1$ are followed.

The idea is as follows: if $P\lor Q$ is proven, then an agent that holds a stock each in sentences $P$ and $Q$ must necessarily be regarded as having \emph{at least} \$1 between these assets. Similarly if $\lnot P\lor\lnot Q$ is proven, then that agent has \emph{at most} \$1 between these assets. We call these different logical possibilities ``worlds'' or rather ``worlds that are propositionally consistent with the output of the theorem prover so far'' (PC worlds for short). In particular when evaluating whether a trade is within budget (to be accepted by the market-maker), we demand that it is within budget in all PC worlds. More precisely:

\begin{definition}[Worlds and valuations]
    Let $\inventory:\Props\too\Rats$ be a finite-supported map, and let $\world:\Props\to\Bools$ be a \emph{world}, i.e. a truth-assignment, so the dot product $\world\cdot\inventory$ represents the valuation of $\inventory$ according to $\world$.
\end{definition}

\begin{definition}[PC worlds]
     A world is said to be propositionally consistent (PC) if for all $P\in\Props$, $\world(P)$ is determined by Boolean algebra from the prime sentences in $\Props$, i.e. $\world(P\land Q)=\world(P)\land \world(Q)$, $\world(P\lor Q)=\world(P)\lor \world(Q)$, etc. Furthermore, let $\thms_t$ be the subset of $\Props$ proven by time $t$: then a world is said to be PC with $\thms_t$ if (1) it is PC and (2) $\world(P)=1$ for all $P\in\thms_t$. We denote the set of worlds PC with $\thms_t$ as $\PC(\thms_t)$; in particular $\PC:=\PC(\{\})$. For any $\inventory$, denote its set of plausible valuations as $\PC(\thms_t,\inventory):=\{\world\cdot\inventory\mid \world\in\PC(\thms_t)\}$.
\end{definition}

Note that $\PC(\thms_t,\inventory)$ is computable, since you only have to check propositional consistency for the sentences actually supported by $\inventory$. 

Another quirk of Garrabrant induction is that agents have type $\agent:\Nats\to(\Props\to\probs)\to(\Props\to\Rats)$ i.e. they output ``joint'' demand schedules allowing for cross-elasticity of demand rather than leaving this to be figured out by the market dynamics. Naturally, calculating equilibrium between such demand schedules is non-trivial and requires Brouwer's fixed point theorem, as well as constraints on $\agent$ (specifically that $\agent(t)$ is continuous in price and comprised only of some particularly simple expressions depending only on some external information like price history). The framework actually implements a rational approximation of the equilibrium computed via a brute-force Farey enumeration of all rational numbers.

\begin{definition}[Garrabrant induction]
    Fix a language $\Props$, a theorem enumerator $\thms:\Nats\to\finset\Props$ (obeying in particular $s\le t\implies\thms_s\subseteq\thms_t$) and enumerator of agents $\father:\Nats\to\agent[s]\times\Rats$ such that $\father_1$ is bijective and $\sum_t\father_2(t)<\infty$ (where $\father=(\father_1,\father_2)$ and $\agent[s]$ is the type specified earlier $\Nats\to(\Props\to\probs)\to(\Props\to\Rats)$). Then the Garrabrant induction algorithm is given by the following mutual recursion:
    \begin{itemize}
        \item an ``aggregate trader'' $$\agent[m](t,\pricee):=\sum_{\agent\in\{\father(1)\dots\father(t)\}}\indicate{\min\PC(\thms_t,\inventory(t,\pricee) + \agent(t,\pricee))\ge 0}\agent(t,\pricee)$$ where $\indicate{}$ denotes an indicator function
        \item an ``equilibrium price'' $\pricee(t)$, which approximates a zero of $\agent[m](t)$ i.e. so that $\agent[m](t,\pricee(t))\approx 0$ (in particular the error should be $\le 1/2^t$)
        \item an inventory account $\inventory(t,P)$, as computed by the following algorithm:
    \end{itemize}
    
    \begin{algorithm}
    \begin{algorithmic}
    \Function{$\inventory$}{$\tau$}
        \For{$t\le\tau$}
            \State $\inventory[\father_1(t)](\top)\gets\father(t)_2$
            \For{$X\in\supp{\inventory}$} \Comment{resolve proven sentences}
                \If{$X\in\thms_t$}
                    \State $\inventory(\top) \gets \inventory(\top)+\inventory(X)$
                    \State $\inventory(X)\gets 0$
                \EndIf
            \EndFor
            \For{$X\in\supp\agent$}
                \If{$\min\PC(\thms_t,\inventory(t,\pricee) + \agent(t,\pricee))\ge 0$}
                    \State $\inventory\gets\inventory+\agent(t)$ \Comment{add trade to inventory}
                \EndIf
            \EndFor
        \EndFor
        \State \Return $\inventory$
    \EndFunction
    \end{algorithmic}
    \end{algorithm}
    \label{def:garrabrant}
\end{definition}

[A slight difference in the formulation we have presented is that we put the onus on the individual agents to calculate and include the cash payment within their orders, and the aggregate agent just zeroes their trade if they submit an invalid trade; this makes no difference to the whole algorithm.]

\begin{theorem}[Inexploitability]
    With definitions as in Def~\ref{def:garrabrant}, \emph{no trader $\agent\in\agent[s]$ can exploit the price sequence $\pricee$}, i.e. $\PC(\thms_t,\inventory(t))$ remains bounded from above. 
    \label{thm:garrabrant}
\end{theorem}

\begin{hproof}
    If any trader could exploit the price sequence, so could the aggregate trader $\agent[m]$ (as whichever trader exploits the market would also grow its proportion of $\agent[m]$'s trades to $+\infty$, and no trader can run $\agent[m]$'s finances to $-\infty$ since it would just go bankrupt). However, $\agent[m]$ cannot exploit $\pricee$ as it is set to be (arbitrary close to) an equilibrium for $\agent[m]$.
\end{hproof}

Indeed, Garrabrant induction may be applied to some FOL theory and be seen as a competing framework to ours -- however, it has crucial limitations:
\begin{itemize}
    \item It does not exploit the \emph{FOL structure} of the language, so it doesn't tell us anything about the relationship between the sentences $P(x)$ and $\exists x, P(x)$ as e.g. probability theory would.
    \item It is fundamentally ``theory-dependent'', which leaves some foundational questions unanswered from a philosophical perspective: it doesn't address the question of \emph{why} we give credence to particular formal theories for forming our beliefs (ideally, our answer to this should look like ``because formal theories make a lot of money on the market''). Although this may seem like a rather detached philosophical problem, it is likely relevant should we seek to design intelligent agents that form beliefs about mathematical truths (for general examples of market-based designs of decision-making agents see \cite{oesterheldTheoryBoundedInductive2023, changDecentralizedReinforcementLearning2020, kweeMarketBasedReinforcementLearning2001, baumModelIntelligenceEconomy1999, balduzziCorticalPredictionMarkets2014}), as we would like the agent to be able to reflect beyond the computational limitations of a particular theory.
    \item It cannot be readily extended beyond the realm of formal logic: while \cite{garrabrantLogicalInduction2016} does make note that the framework performs something akin to empirical induction for sentences independent of a theory, this notion of empirical induction once again does not exploit the FOL structure of the language.
\end{itemize}

\subsection{Bounded rationality}
\label{sec:br}

The relevance of logical uncertainty has also been occassionally remarked on in the \emph{bounded rationality} literature: e.g. by Halpern \cite{halpernAlgorithmicRationalityGame2014, halpernDonWantThink2011} and Oesterheld \cite{oesterheldTheoryBoundedInductive2023} who gave examples of games and decision problems that depended on logical questions. However, in all these examples, the logical questions were only $\Delta_0$, as the games therein have to pay out a reward at a fixed point depending on the answer to the question.

\subsection{Partial approximations}
\label{sec:approx}

An immediately apparent solution to our problem might be to approximate first-order logic sentences with bounded comprehensions -- e.g. for $\exists x,\forall y, P(x,y)$ consider the asset that on day $t$ is equal to (can be exchanged for) the $\Delta_0$ sentence $\exists x\le t,\forall y\le t, P(x,y)$. However, a simple counter-example is provided by $P(x,y):=x\ge y$ (see Fig~\ref{fig:checkerboard} for a simple illustration), and Lemma~\ref{lem:nogo} proves that no such straightforward computable algorithm can address our problem.

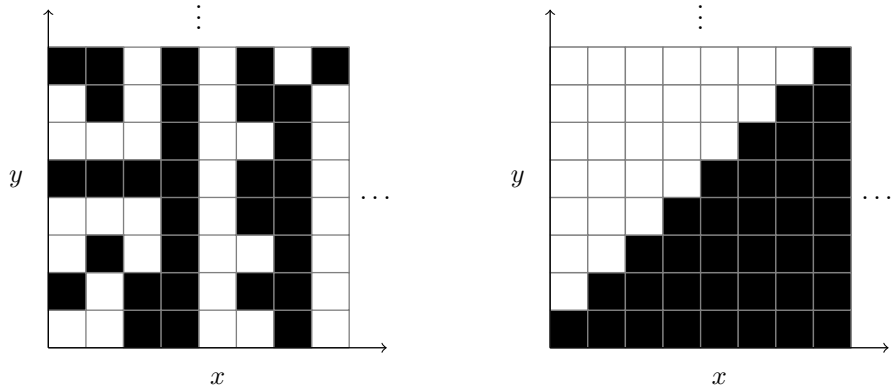
\begin{figure}
    \centering
    \begin{subfigure}{0.45\textwidth}
        \centering
        \usetikzlibrary{arrows}
\usetikzlibrary{positioning}

\begin{tikzpicture}[line width =0.5pt, scale=0.5]

\draw[gray] (0,0) grid (8,8);
\foreach \x/\y in {0/1, 0/4, 0/7, 1/2, 1/4, 1/6, 1/7, 2/0, 2/1, 2/4, 3/0, 3/1, 3/2, 3/3, 3/4, 3/5, 3/6, 3/7, 5/1, 5/3, 5/4, 5/6, 5/7, 6/0, 6/1, 6/2, 6/3, 6/4, 6/5, 6/6, 7/7}
    \filldraw[fill=black, draw=gray] (\x,\y) rectangle ++(1,1);
\draw[->] (0,0) -- (9,0) node[midway, below=6pt]{$x$};
\draw[->] (0,0) -- (0,9) node[midway, left=6pt]{$y$};
\node[] at (4,9){$\vdots$};
\node[] at (8.75,4){$\dots$};
\end{tikzpicture}
        \label{fig:checkerboard1}
    \end{subfigure}\hfill%
    \begin{subfigure}{0.45\textwidth}
        \centering
        \usetikzlibrary{arrows}
\usetikzlibrary{positioning}

\begin{tikzpicture}[line width =0.5pt, scale=0.5]

\draw[gray] (0,0) grid (8,8);
\foreach \x in {0, 1, 2, 3, 4, 5, 6, 7}
    \foreach \y in {0, 1, 2, 3, 4, 5, 6, 7}
        \ifthenelse{\x<\y}{\def\col{white}}{\def\col{black}}
        \filldraw[fill=\col, draw=gray] (\x,\y) rectangle ++(1,1);

\draw[->] (0,0) -- (9,0) node[midway, below=6pt]{$x$};
\draw[->] (0,0) -- (0,9) node[midway, left=6pt]{$y$};
\node[] at (4,9){$\vdots$};
\node[] at (8.75,4){$\dots$};
\end{tikzpicture}
        \label{fig:checkerboard2}
    \end{subfigure}
    \caption{Shading $P(x,y)$ as a checkerboard, the $\Sigma_2$ sentence $\exists x,\forall y, P(x,y)$ can be interpreted as ``there is an infinite black column''. Left: an arbitrary example; Right: the example of $P(x,y):=x\ge y$, where the $\Sigma_2$ sentence is false, yet it is true for every ``finite window''.}
    \label{fig:checkerboard}
\end{figure}

\begin{lemma}[the problem is not trivial]
    An asset mechanism is a computable real-valued function $v:\Nats\times\Props\to\mathbb{R}$ such that $\lim_{t\to\infty}v(t,P)=1$ if $P$ and $0$ if $\lnot P$. A scoring rule mechanism is a computable real-valued function $s:\Nats\times\Props$ such that $\lim_{t\to\infty}s(t,P,p)=\log(p)$ if $P$ and $\log(1-p)$ if $\lnot P$. If $\Props$, the class of propositions considered, includes at least $\Sigma_4$ or $\Pi_4$ sentences, neither an asset mechanism nor a scoring rule mechanism exists. 
    \label{lem:nogo}
\end{lemma}
\begin{proof}
    The statements $\lim_{t\to\infty}v(t,P)=1$ and $\lim_{t\to\infty}s(t,P,p)=\log(p)$ are both $\Pi_3$ for any $P, p$; thus for them to be equivalent to $P$ for all $P$ would violate Tarski's theorem.
\end{proof}

Instead in this paper, we propose a game-theoretic solution relying on Hintikka's \emph{game semantics} \cite{hintikkaGameTheoreticSemantics1997}. Naturally by Lemma~\ref{lem:nogo} this does not have the standard property expected from a traditional prediction market (that agents are incentivized to bet the price up to their subjective probability for the sentence), however it does have a closely related property that may be interpreted in terms of \emph{constructivism} in the philosophy of mathematics, as we will see.

\begin{notation*}
    Booleans are denoted as $\Bools$ representing ``false'' and ``true'' respectively. Multivariable functions may be denoted as $f:S_1\to S_2\dots S_n\to T$ i.e. with $\to$ right-associative; we may sometimes type a function $f:S\to T$ as $f:(s:S)\mapsto T$ or even $f:s\mapsto T$ like a sort of $\lambda$-notation; wherein unless otherwise specified, time denoted by $t$ belongs to $\Nats$. $\proj[TU]:S\times T\times U\to T\times U$ denotes projection; $\finset{S}$ denotes the set of finite sets of elements in $S$; for a subset $T\subseteq S$, $\compl{T}$ denotes its complement; if a set $S$ contains an element $\zero$, then $f:T\too S$ denotes a function with finite support $\supp{f}:=\preimg{f}{\compl{\{\zero\}}}:= \{x\mid f(x)\ne\zero\}$; the addition of a canonical $\zero$ element to a set $S$ is denoted by $\maybe{S}:=S\cup\{\zero\}$. We use the notation $f:T\toc S$ to denote an arbitrary partial computable function, $f:T\topt S$ to denote a function that is polynomial-time in input $t$, $f:T\tod S$ to denote a non-increasing function, and $f:T\toe S$ to denote a piecewise-constant function over a finite number of pieces, i.e. an $S$-labeled partition of $T$ (so e.g. $f:T\tode S$ denotes a function that is both non-increasing and piecewise-constant), and $\intlen{l}$ to denote the length of a clopen interval or finite union of clopen intervals $l\subseteq\Rats$. Denote by $\strings$ the set of finite strings with the infix $\concat{l_1}{l_2}$ indicating string concatenation, $\Props$ the set of FOL sentences in prenex normal form, $\Props^+:=\Props\cap\compl{\Delta_0}$ and $\Sigma=\bigcup\Sigma_n$, $\Pi=\bigcup\Pi_n$. For $n\in\Nats$ and $P\in\Props^+$, the replacement of the leading variable in $P$ by $n$ is denoted $\play{P}{n}$ (i.e. if $P:=\exists x,p(x)$ or $P:=\forall x,p(x)$ then $\play{P}{n}:=p(n)$). Let $N\in\moves$: if $N=\zero$, $\play{P}{N}:= P$; else if $P\in\Sigma$, then $\play{P}{N}:=\bigvee_{n\in N}\play{P}{n}$; else if $P\in\Pi$, then $\play{P}{N}:=\bigwedge_{n\in N}\play{P}{n}$, all reduced to prenex normal form. We may also use this notation for a map $\instaplayer:\Props\to\finset\Nats$, i.e. denote $\play{P}{\instaplayer}:=\play{P}{\instaplayer(P)}$, and represent successive applications as $\play{P}{M,N}:=\play{\play{P}{M}}{N}$, etc. likewise $\play{P}{\instaplayer,\instaplayer[b]}:=\play{\play{P}{\instaplayer}}{\instaplayer[b]}$ etc.
\end{notation*}

\section{Framework}
\label{sec:framework}

We present a slightly modified version of Hintikka's Verification-Falsification (VF) game \cite{hintikkaGameTheoreticSemantics1997}, the crucial difference being that players' moves are in $\finset{\Nats}$ rather than $\Nats$ itself.

\begin{definition}[Verification-Falsification game]
    To each FOL sentence $P$ we associate a game between two players, the ``Verifier'' and the ``Falsifier''. If $P$ is $\Delta_0$, then the Verifier or Falsifier win \$1 if $P$ is true or false respectively. Else if $P$ is $\Sigma$, the Verifier goes first and if $P$ is $\Pi$, the Falsifier goes first. The first player chooses $S\in\finset{\Nats}$ as their first move, and the game proceeds as the VF game for $\play{P}{S}$.
    \label{def:vfgame}
\end{definition}

The idea is then as follows: \emph{we measure how much traders are willing to pay to play the Verification-Falsification game for some sentence $P$}. Another way of expressing this is: the asset for some sentence $\exists x,P(x)$ is the \emph{option} to exchange it for any $\bigvee_{x\in S}P(x)$ for $S$ finite; while $\forall x,P(x)$ is the \emph{obligation} to exchange it for any $\bigwedge_{x\in S}P(x)$ chosen by the opponent.

Beyond this generic description, we can give a specific construction of a \emph{program market}: a market whose participants are programs belonging to some specific class, which trade on (and in our case, play the VF game for) sentences. Such models have been described in e.g. \cite{garrabrantLogicalInduction2016, oesterheldTheoryBoundedInductive2023}; in line with them, we let the participants be polynomial-time agents. 

Def~\ref{def:vfmarket} makes our construction precise: note that this is a very theoretical framework, intended to prove that certain optimality results hold ``at least in principle''. Practically, our framework can be implemented simply by a continuous double-auction mechanism.

\begin{definition}[Prediction Market for VF games]
    Define the following:
    \begin{itemize}
        \item A \emph{price setter} $\agent[p]$ is composed of:
        \begin{itemize}
            \item a \emph{price sequence} $\price[p]:t\mapstopt\Props\too\probs$
            \item a \emph{player} $\player[p]:t\mapstopt\Props\too\strings\too\moves$
            \item a \emph{labeler} $\labeler[p]:t\mapstopt\Props\too\Rats\toe\maybe{\strings}$ such that $\intlen{\supp{\labeler[p](t,P)}}=\inventory[m](t,P)$ (defined by mutual recursion below)
        \end{itemize}
        \item An \emph{agent} $\agent$ is composed of:
        \begin{itemize}
            \item an \emph{endowment} $\endowment\in\Rats$ and a \emph{birthday} $\birthday\in\Nats$
            \item a \emph{trader} $\trader:t\mapstopt\Props\too\probs\tode\Rats$
            \item a \emph{player} $\player:t\mapstopt\Props\too\strings\too\moves$
            \item a \emph{labeler} $\labeler:t\mapstopt\Props\too\Rats\toe\maybe{\strings}$ such that $\intlen{\supp{\labeler(t,P)}}=\inventory(t,P)$ (defined by mutual recursion below)
            \item an \emph{inventory} $\inventory:t\mapstopt \Props\too\Rats$ defined by (initial conditions) $\inventory(s,\top)=0$ for $s<\birthday$, $\inventory(\birthday,\top)=\endowment$ and for all other propositions $\inventory(0,P)=0$ and (recursion rule) $\inventory(t+1)=\inventory(t)+\sum_{P,i}\change[P,i](t)$ where:
            \begin{itemize}
                \item orders placed: $\change[P,1](t,P)=\lambda\trader(t,P,\price(t,P))$ where $\lambda$ is 1 if the following conditions are met (else 0): for all $P$, $\max_p-\trader(t,P,p)\le \inventory(t,P)$ (you're not selling what you don't have) and $\sum_P\max_p p\trader(t,P,p)\le\inventory(t,\top)$ (you can afford all your purchases)
                \item cost of orders placed: $\change[P,2](t,\top)=-\price(t,P)\change[P,1](t,P)$
                \item the moves played (if $P\in\Sigma$): $\change[P,3](t,P)=-\sum_{l\in\supp{\player(t,P)}}{\intlen{\preimg{\labeler(t,P)}{l}}}$ and $\change[P,4](t,\play{P}{\player(t,P,l)})=\intlen{\preimg{\labeler(t,P)}{l}}$
                \item the opponent's moves (if $P\in\Pi$): $\change[P,5](t,P)=-\sum_{l\in\supp\player(t,P)}{\intlen{\preimg{\labeler[p](t,P)}{l}}}$ and $\change[P,6](t,\play{P}{\player[p](t,P,l)})=\intlen{\preimg{\labeler[p](t,P)}{l}}$
                \item the payout from empirical truth (if $P\in\Delta_0$): if $P\in\supp\xi(t)$, then $\change[P,7](t,P)=-\inventory(t,P)$ and $\change[P,8](t,\xi(t,P))=\inventory(t,P)$
            \end{itemize}
        \end{itemize}
        \item the \emph{empirical reality} is a process $\xi:t\mapstopt\Delta_0\too\maybe{\Bools}$ such that $s\le t,P\in\supp\xi(s)\implies\xi(t,P)=\xi(s,P)$ and $\bigcup_t\supp\xi(t)=\Delta_0$ and $\xi(t,\lnot P)=\lnot\xi(t,P)$
        \item The type of agents is denoted as $\agent[s]:=\endowment[s]\times\birthday[s]\times\trader[s]\times\player[s]\times\labeler[s]$ where the respective types are already as specified (including sub-typing by the requisite condition); define specifically $\tpler:=\trader[s]\times\player[s]\times\labeler[s]$.
        \item The \emph{father of agents} is an enumerator and allocator for agents, i.e. a map $\father:t\mapsto\agent[s]$ such that $\proj[\tpler]\circ\father$ is bijective; the total endowment is finite $\sum_t{\endowment[\father(t)]}<\infty$; and the birthdays are correct $\birthday[\father(t)]=t$. It is associated with an \emph{aggregate agent} as follows:
        \begin{itemize}
            \item $\endowment[m]=\sum_t{\endowment[\father(t)]}$ and $\birthday[m]=0$
            \item $\trader[m](t)=\sum_{\agent=\father(1),\dots\father(t)}\lambda_{\agent}\trader(t)$ where $\lambda_{\agent}$ is 1 if the following conditions are met (else 0): for all $P$, $\max_p-\trader(t,P,p)\le\inventory(t,P)$ and $\sum_P{\max_{p}{p\trader(t,P,p)}}\le\inventory(t,\top)$ (this sum is only a finite sum, because only finitely many agents have any cash at $t$)
            \item $\player[m](t,P,\concat{\code{\agent}}{l})=\player(t,P,l)$ where $\code{\agent}$ indicates some encoding for agents
            \item $\labeler[m](t,P,q)=
            \begin{cases}
                \concat{\father(t_q)}{\labeler[\father(t_q)](t,P)} & \text{if }t_q\le t \\
                \zero & \text{else}
            \end{cases}$ \\
            where $t_q$ is the smallest value such that $\sum_{s\le t_q}{\inventory[\father(s)](t,P)}\ge q$, if it exists, else $\infty$.
        \end{itemize}
        \item A special price setter $\agent[pp]$, called \emph{equilibrium}, is defined as follows -- for each $t$ and $P$:
        \begin{itemize}
            \item $\price[pp](t,P)$ is any solution $p$ to $\trader[m](t,P,p)-\trader[m](t,\lnot P,1-p)=0$.
            \item $\player[pp](t,P)=\player[m](t,\lnot P)$
            \item $\labeler[pp](t,P)=\labeler[m](t,\lnot P)$
        \end{itemize}
    \end{itemize}
    \label{def:vfmarket}    
\end{definition}

A more readable description: each agent manages its \emph{inventory} $\inventory(t):\Props\too\Rats$; in particular $\inventory(t,\top)$ denotes an agent's cash reserves. At each point in time, $\trader(t,P)$ denotes its demand (or supply, if negative) schedule for $P$, which is a sum of limit orders. The description of the player has some subtle considerations: (1) the moves are in $\moves$ rather than $\finset\Nats$, because the move $\zero$ is interpreted as ``pass'', i.e. should the agent wishes to not instantly play his move when he acquires $P$ but compute his move over several time-steps (2) an agent that holds multiple stocks of some sentence might not want to play the same move on all of them.

We capture this behaviour by having the agent \emph{labelling} different portions of his inventory of stock $P$ with different labels, i.e. a labelled partition $\Rats\toe\maybe{\strings}$ of $\Rats$ so that only a fragment of $\Rats$ whose length equals the agent's stock in $P$ is mapped to a label, everything else is mapped to $\zero$. When $\player$ then plays a move for sentence $P$ and label $l$, the amount of $P$ stocks labelled with $l$ will be removed from the inventory and replaced with $\play{P}{\player(t,P,l)}$. 

In particular, the aggregate agent $\agent[m]$ uses this system to keep the inventories for different agents separate by prefixing their labels with some code for each agent as $\concat{\code{\agent}}{l}$ -- this is necessary for the crux of the program market concept, which is that each agent can only trade with its own money, so more successful agents gain influence (in the form of wealth) while those that go bankrupt are effectively removed from the market. This is done through the indicator denoted by $\lambda_{\agent}$, which checks if $\agent$ can afford the trades it makes.

Finally, some notes regarding equilibrium calculation: in our framework, agents provide independent demand schedules for each sentence (i.e. a map $\Props\too(\probs\to\Rats)$, rather than a map $(\Props\too\probs)\to\Rats$): in other words, should the agents wish for cross-elasticity in their demand-schedules, the onus is on them to estimate the prices of other sentences. This reduces equilibrium calculation to finding the zero of a non-increasing piecewise-constant function, which is elementary. The equilibrium price setter's player is more subtle: at equilibrium price, $\agent[m]$ buys equal amounts of $P$ and $\lnot P$, however it may use different players for different portions of each. So we have $\player[pp]$ use $\agent[m]$'s players for $\lnot P$ against $\agent[m]$'s players for $P$ and vice versa so it wins exactly as many games as it loses, as illustrated in Fig~\ref{fig:matchup}.

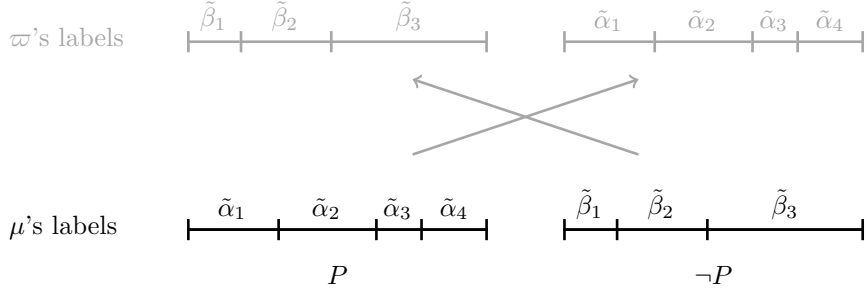
\begin{figure}
    \centering
    \usetikzlibrary{arrows}
\definecolor{aqaqaq}{rgb}{0.65,0.65,0.65}
\begin{tikzpicture}[line width =1pt]%

\draw (-2.5,0.32) node[anchor=north west] {$\mu\text{'s labels}$};

\draw (0,0) -- (4,0) node[midway, below = 10pt]{$P$};
\draw [|-] (0,0) -- (1.2,0) node[midway, above]{$\small\tilde{\alpha}_1$};
\draw [|-] (1.2,0) -- (2.5,0) node[midway, above]{$\small\tilde{\alpha}_2$};
\draw [|-] (2.5,0) -- (3.1,0) node[midway, above]{$\small\tilde{\alpha}_3$};
\draw [|-|] (3.1,0) -- (4,0) node[midway, above]{$\small\tilde{\alpha}_4$};

\draw (5,0) -- (9,0) node[midway, below = 10pt]{$\lnot P$};
\draw [|-] (5,0) -- (5.7,0) node[midway, above]{$\small\tilde{\beta}_1$};
\draw [|-] (5.7,0) -- (6.9,0) node[midway, above]{$\small\tilde{\beta}_2$};
\draw [|-|] (6.9,0) -- (9,0) node[midway, above]{$\small\tilde{\beta}_3$};

\begin{scope}[color = aqaqaq]
\draw (-2.5,2.82) node[anchor=north west] {$\varpi\text{'s labels}$};

\draw [|-] (0,2.5) -- (0.7,2.5) node[midway, above]{$\small\tilde{\beta}_1$};
\draw [|-] (0.7,2.5) -- (1.9,2.5) node[midway, above]{$\small\tilde{\beta}_2$};
\draw [|-|] (1.9,2.5) -- (4,2.5) node[midway, above]{$\small\tilde{\beta}_3$};

\draw [|-] (5,2.5) -- (6.2,2.5) node[midway, above]{$\small\tilde{\alpha}_1$};
\draw [|-] (6.2,2.5) -- (7.5,2.5) node[midway, above]{$\small\tilde{\alpha}_2$};
\draw [|-] (7.5,2.5) -- (8.1,2.5) node[midway, above]{$\small\tilde{\alpha}_3$};
\draw [|-|] (8.1,2.5) -- (9,2.5) node[midway, above]{$\small\tilde{\alpha}_4$};

\draw [->] (6,1) -- (3,2);
\draw [->] (3,1) -- (6,2);
\end{scope}

\end{tikzpicture}
    \caption{Illustration of $\varpi$'s player countering $\mu$}
    \label{fig:matchup}
\end{figure}

\section{Results}
\label{sec:results}

\begin{definition}[Exploitation]
    An agent $\agent$ is said to \emph{exploit} a price-setter $\agent[p]$ if $\{\inventory(t,\top):t\in\Nats\}$ is bounded from below but not bounded from above.
    \label{def:exploit}
\end{definition}

\begin{lemma}[Inexploitabity]
    There is no agent $\agent\in\agent[s]$ that exploits $\agent[pp]$.
    \label{lem:inexploit}
\end{lemma}
\begin{proof}
    If any $\agent\in\agent[s]$ exploits $\agent[pp]$, then so does $\agent[m]$ (because the inventory allocated to $\agent$ will increase without bound, while all other agents' inventories are bounded below by 0 since they can only spend from their alloted inventory). However, by construction, $\agent[m]$ does not exploit $\agent[pp]$ (as it always holds an equal number of $P$ and $\lnot P$ stocks, and wins exactly as many games as it loses). Thus, no $\agent\in\agent[s]$ exploits $\agent[pp]$.
\end{proof}

\begin{theorem}[Convergence]
    $\lim_{t\to\infty}\price[pp](t,P)$ exists for all $P$; denote this as $\longrun(P)$.
    \label{thm:convergence}
\end{theorem}
\begin{proof}
    Suppose it didn't; then there exists $x\in(0,1)$ and $\varepsilon>0$ such that $\price[pp](t,P)>x+\varepsilon$ infinitely often and $\price[pp](t,P)<x-\varepsilon$ infinitely often. Then consider an agent given by a trader that sells when $\price[pp](t,P)>x+\varepsilon$ and buys when $\price[pp](t,P)<x-\varepsilon$, and a trivial player (doesn't play at all; returns $\zero$ each time). This agent exploits the market.
\end{proof}

Ideally, we would like to say that our market learns to correctly price very FOL sentence: that $P$ is equivalent to $\longrun(t,P)=1$, and $\lnot P$ is equivalent to $\longrun(t,P)=0$: because if a true sentence did not approach \$1, an agent could buy arbitrarily large quantities of it and win the VF game on them. However, we know from Lemma~\ref{lem:nogo} that this is impossible: indeed, the flaw in our intuition is that the mere existence of a winning strategy (i.e. the ``truth'' of $P$) does not imply it is actually computable. Instead, we adopt the following ``constructivist'' notion of truth, closely related to the notion of ``CGTS-truth'' (``Computational Game Theoretic Semantics Truth'') proposed in \cite{boyerProofTruth2012b}.

\begin{definition}[Constructive truth]
    Denote $\instaplayer[s]:=\Props\toc\finset\Nats$. For any $P\in\Props$ and $\instaplayer,\instaplayer[b]\in\instaplayer[s]$, observe that the sequence $\play{P}{\instaplayer,\instaplayer[b],\instaplayer,\instaplayer[b],\dots}$ eventually converges to a $\Delta_0$ sentence: denote this sentence by $\playfull(P,\instaplayer,\instaplayer[b])$. Then an FOL-sentence $P$ is said to be $\instaplayer[s]$-true if $\exists\instaplayer\in\instaplayer[s],\forall\instaplayer[b]\in\instaplayer[s],\playfull(P,\instaplayer,\instaplayer[b])$, and $\instaplayer[s]$-false if $\exists\instaplayer[b]\in\instaplayer[s],\forall\instaplayer[a]\in\instaplayer[s],\lnot\playfull(P,\instaplayer,\instaplayer[b])$.
    \label{def:constructive}
\end{definition}

\begin{lemma}[Correspondence between $\instaplayers$ and $\players$]
    For any $\instaplayer:\Props\toc\finset\Nats$ there is a corresponding $\player:t\mapstopt\Props\too\strings\too\moves$ that executes it,, such that the following conditions are met: (1) $\forall s\ne t,\supp\instaplayer(t)\cap\supp\instaplayer(s)=\varnothing$ and (2) if $P\in\supp\instaplayer$ ($\instaplayer$ halts on input $P$) then $\exists t,\player(t,P,\foo)=\instaplayer(P)$ else $\forall t,\player(t,P,\foo)=\zero$.
    \label{lem:correspondence}
\end{lemma}
\begin{proof}
    Fix an enumeration of $\Props$ i.e. a bijection $\zeta:\Nats\to\Props$, and fix a model of computation for $\instaplayer$, e.g. a Turing Machine. Then define $\player$ as follows:
    \begin{equation*}
        \player(t,P,l)=
        \begin{cases}
            \instaplayer(P) & \text{if } P\in\{\zeta(0),\dots,\zeta(t)\},\\&\ \ \ %
            \halts(\instaplayer,P,t),\\&\ \ \ \:%
            l=\foo \\
            \zero & \text{else}
        \end{cases}
    \end{equation*}
\end{proof}

\begin{theorem}[Learning constructive truth]
    If $P$ is an $\instaplayer[s]$-true sentence, then $\longrun(P)=1$. 
    \label{thm:main}
\end{theorem}
\begin{proof}
    Suppose it wasn't; then there is some $\varepsilon$ such that $\price[pp](t,P)$ is below $1-\varepsilon$ infinitely many times. Then consider the agent given by (1) the trader $\trader$ that buys whenever $\price[pp](t,P)<1-\varepsilon$ (2) the player $\player$ that is the Lemma~\ref{lem:correspondence}-correspondent of the $\instaplayer$ affirmed in the hypothesis $\exists\instaplayer\in\instaplayer[s],\forall\instaplayer[b]\in\instaplayer[s],\playfull(P,\instaplayer,\instaplayer[b])$ (3) the labeler $\labeler$ that returns $\foo$ on all outputs. This agent exploits the market.
\end{proof}

\begin{corollary}[Learning constructive falsehood]
    If $P$ is an $\instaplayer[s]$-false sentence, then $\longrun(P)=0$.
    \label{cor:main}
\end{corollary}
\begin{proof}
    Apply Thm~\ref{thm:main} to $\lnot P$.
\end{proof}

\section{Discussion}

We have presented a prediction market on sentences expressible in First-Order Logic and proven that our framework learns a notion of ``$\instaplayer[s]$-truth'' in the form presented in Def~\ref{def:constructive}. Although usually applied to mathematical logic, First-Order Logic can also express sentences about some empirical universe, even if the process that generates empirical truth (denoted $\xi$ in our framework) is probabilistic in nature.

However, our result (Thm~\ref{thm:main}) is still somewhat weak, in that it says nothing about sentences that are neither $\instaplayer[s]$-true nor $\instaplayer[s]$-false -- from Thm~\ref{thm:convergence}, we know that the market price of every sentence converges to \emph{something}, and it would be interesting to study the nature of this limiting distribution. It would not be a probability distribution in the traditional sense -- i.e. you would not have $\Prob[\exists x,P(x)]=\sup_x\Prob[\exists i\le x,P(i)]$ and $\Prob[\forall x,P(x)]=\inf_x\Prob[\forall i\le x,P(i)]$, as this would contradict Lemma~\ref{lem:nogo}. Instead we might consider an alternate definition of a $\sigma$-algebra which is required to be closed only under unions and intersections of \emph{computable sequences}, i.e. replacing the countable union axiom with ``$\psi:\Nats\toc\mathcal{F}\implies\bigcup_i\psi(i)\in\mathcal{F}$'' and adopting a suitable generalized notion of probability measure on this algebra. Possibly relevant work to this end includes: quantifier algebras \cite{cwooQuantifierAlgebra2013}, Shafer \& Vovk's game-theoretic formulation of probability theory \cite{shaferProbabilityFinanceIt2005, shaferGameTheoreticFoundationsProbability2019} and Japaridze's ``Computability Logic'' \cite{japaridzeBeginningWasGame2009, japaridzeSurveyComputabilityLogic2015}.

Our work has a range of applications: the straightforward ones are to the study of logical uncertainty and relatedly bounded rationality, and to designing actual practical prediction markets for such sentences (e.g. prediction markets for various practically relevant statistical questions may be expressed in terms of limiting distributions of experimental results, thus as $\Pi_3$ sentences). With practical markets, however, special attention must be taken to take into account asymmetries in the computational costs of players for opposing sides: in our theoretical framework, every possible player was enumerated in the market, however practically some players may be ``more expensive'' than others, and this may influence how much traders are willing to pay for a side. For example, this may be addressed by ``separating out'' the market for players from the prediction market, so that traders are explicitly paying a price for purchasing players and this price can be added to the probability estimate for a sentence. In particular, this might also make it feasible to implement such systems with an automated market maker (e.g. logarithmic market scoring \cite{hansonLogarithmicMarketScoring2002, hansonCombinatorialInformationMarket2003}), as the obstacle to doing this immediately with our framework is that it's not obvious what player such a market-maker would use. 

There are two other potential implications of our work deserving of attention.

Firstly, in formal logic, where it provides a new approach to measuring the strength of a mathematical theory: in a slightly modified version of game semantics that allows one player to ``go back and change its moves'' \cite{bonnayPreuvesJeuxSemantiques2004}, a formal proof from a mathematical theory may be understood as a computable strategy to this modified verification-falsification game \cite{boyerProofTruth2012b}; we my then consider the maximum wealth that can be acquired by an agent that only trades and plays in accordance with a theory, and regard this as a measure of the strength of the theory.

Second, our work serves as a toy example or entrypoint toward a much more ambitious research problem: \emph{betting on sentences in the latent space}. This problem has been discussed e.g. by \cite{tailcalledLatentVariablesPrediction2023, christianoElicitingLatentKnowledge2021, baillonBayesianMarketsElicit2017, baillonSimpleBetsElicit2021} wherein partial approaches have been highlighted; briefly: there is a vast category of sentences (commonly referred to as ``subjective'') that have no obvious basis whatsoever in empirical reality, but yet we consider meaningful, because they provide information on (i.e. are correlated with) things that do have some basis in empirical reality -- e.g. a historical claim, like ``Bob committed the crime'', will not in itself be revealed by any future empirical observation, but it may be correlated with other, more useful sentences like ``evidence of Bob's wrongdoings will be found'', ``if Bob is imprisoned, lives will be saved'' or ``if you hire Bob, you might find office supplies missing''. One may say that these useful sentences are all correlated, and therefore ``Bob committed the crime'' is a useful fiction constructed to capture the signal in the noisy data -- i.e. \emph{a variable in the latent space}. A comprehensive framework to address bets on such variables -- i.e. a market that decides what variables should exist in the latent space, then provides arbitration for betting on them -- would serve essentially the same function as an intelligent agent. 

\printbibliography

\end{document}